\documentclass[11pt]{article}
\usepackage[utf8]{inputenc}

\usepackage{titling}
\predate{}
\postdate{}

\usepackage{graphicx} 
\usepackage[backend=bibtex,giveninits=true,maxbibnames=10]{biblatex}

\bibliography{gittins-refs}
\usepackage{subfig}
\usepackage[normalem]{ulem}
\usepackage{float}
\usepackage[margin=1in]{geometry}

\usepackage{amsmath,amsthm,amssymb,latexsym,amsfonts, dsfont}

\usepackage{setspace}
\usepackage{xcolor}

\usepackage{hyperref}
\usepackage{lineno}
\usepackage{amsthm}
\usepackage{bm}
\usepackage{listings}
\usepackage{xcolor,colortbl}
\usepackage{cleveref}
\usepackage{comment}
\usepackage{xspace}
\usepackage{thm-restate}

\newtheorem{algorithm}{Algorithm}

\newtheorem{claim}{Claim}

\newtheorem{definition}{Definition}
\newtheorem{example}{Example}

\newtheorem{lemma}{Lemma}

\newtheorem{proposition}{Proposition}

\newcommand{\pr}{\mathbb{P}}
\newcommand{\dis}{\mathcal{D}}

\newcommand{\gittins}{\textsf{GIPP}\xspace}
\newcommand{\rg}{\textsf{RG}\xspace}
\newcommand{\opt}{\textsf{OPT}}

\AtBeginDocument{%
  }

\title{Robust Gittins for Stochastic Scheduling}

\date{}

\author{Benjamin Moseley\thanks{Carnegie Mellon University. B. Moseley and H. Newman were supported in part by a Google Research Award, an Infor Research Award, a Carnegie Bosch Junior Faculty Chair, NSF grants CCF-2121744 and CCF-1845146 and ONR Grant N000142212702.} \and Heather Newman$^*$ \and Kirk Pruhs\thanks{University of Pittsburgh. Supported in part by NSF grant CCF-2209654.} \and Rudy Zhou\thanks{Microsoft, Supply Chain Optimization Technologies}}

\begin{document}
\maketitle

\begin{abstract}
     A common theme in stochastic optimization problems is that, theoretically, stochastic algorithms need to ``know'' relatively rich information about the underlying distributions. This is at odds with most applications, where distributions are rough predictions based on historical data. Thus, commonly, stochastic algorithms are making decisions using imperfect predicted distributions, while trying to optimize over some unknown true distributions.
    
    We consider the fundamental problem of scheduling stochastic jobs preemptively on a single machine to minimize expected mean completion time in the setting where \textit{the scheduler is only given imperfect predicted job size distributions}. If the predicted distributions are perfect, then it is known that this problem can be solved optimally by the Gittins index policy. 
    
    The goal of our work is to design a scheduling policy that is robust
    in the sense that it produces nearly optimal schedules even if there are modest 
    discrepancies between the predicted distributions and the underlying real distributions. Our main contributions are:  
    \begin{itemize}
        \item We show that the standard Gittins index policy is \textit{not robust} in this sense. If the true distributions are perturbed by even an arbitrarily small amount, then running the Gittins index policy using the perturbed distributions can lead to an unbounded increase in mean completion time.
        \item We explain how to modify the Gittins index policy to make it robust, 
        that is, to produce nearly optimal schedules, where the approximation 
     depends on a new measure of error between the true and predicted distributions that we define.
    \end{itemize}
  Looking forward, the approach we develop here can  be applied more broadly to many other stochastic optimization problems to better understand the impact of mispredictions,  and lead to the development of new algorithms that are robust against such mispredictions.
\end{abstract}

\section{Introduction} \label{sec: intro}

In many (and probably most) stochastic optmization applications, the distributions reported to the algorithm are predictions derived from past observations, and there is no expectation that the predicted distributions are perfectly accurate. 
Yet, the vast majority of the voluminous literature on stochastic 
optimization, or at least the literature on stochastic  scheduling, assumes that the reported/predicted distributions
are perfectly accurate. 
Further, many standard stochastic optimization methods are brittle, 
and in particular not robust against small inaccuracies in the predicted distributions.
Thus, there is a need to develop more robust stochastic optimization algorithms (see, for example, \cite{Asi2019TheIO}). 
Here we contribute to this need of making stochastic scheduling -- and stochastic optimization more generally -- more robust.  In particular, we address a variant of an open problem\footnote{The setting in \cite{scully_et_al:LIPIcs.ITCS.2022.114} is the M/G/1 queue, which is a queue with Poisson arrivals and i.i.d. job sizes. On the other hand, in our work, we consider the no-arrival settings (finitely many jobs all available at time $0$) with independent (but not necessarily identically distributed) job sizes.} stated by Scully, Grosof, and Mitzenmacher \cite{scully_et_al:LIPIcs.ITCS.2022.114},
that is, whether one can design and analyze a  nonanticipatory scheduling policy that is robust with respect to minor errors in the reported probability distributions on the job sizes for the 
following classical stochastic scheduling problem:

\begin{quote}
{\bf  Scheduling Problem Definition:} 
The input consists of non-negative probability distributions $\dis_j$ for $j \in [n]$, where the $j$th job has size $P_j \sim \dis_j$. We assume that the $P_j$'s are independent.

Our goal is to construct a \emph{nonanticipatory} scheduling policy that \emph{preemptively} schedules all $n$ jobs to completion. That is, at any time $t \geq 0$, the scheduling policy selects the job to run at time $t$. A job completes if it has been run for $P_j$ units of time (not necessarily consecutively). If a job $j$ has run for $q$ units of time before time $t$ without completing, the scheduling policy knows that the probability distribution on the size of job $j$ is now $\mathcal{D}_j$ conditioned on $P_j > q$. Since preemption is allowed, the policy can switch between jobs at any time.

The objective is then to minimize the \textit{expected} total completion time, $\mathbb{E}[\sum_{j \in [n]} C_j]$ (or equivalently, mean completion time, $\mathbb{E}[\frac{1}{n} \cdot \sum_{j \in [n]} C_j]$), where $C_j$ is the time that job $j$ completes with respect to a given scheduling policy.
\end{quote}

We emphasize that we study this problem from the perspective of \emph{worst-case analysis}, so we make no assumptions on the job size distributions other than independence. Further, the scheduler is initially only given the \emph{distributions}, $\dis_j$, and observes the realized processing times $P_j$ over time as the jobs are scheduled.

This problem is a stochastic version of the
scheduling problem $1 \mid pmtn \mid \sum_j C_j$ using the standard 3-field scheduling notation~\cite{graham1979optimization}.
Many variations of this stochastic scheduling problem have been studied for decades (good overviews can be found in Pinedo~\cite{Pinedo2022} and  Megow and  Vredeveld \cite{MegowV14}).
For the moment, it is sufficient for our purposes to know the following facts. 
The nonanticipatory policy 
Shortest Expected Processing Time (SEPT), which always schedules the unfinished job whose
expected processing time is minimum, is optimal  for  distributions with increasing hazard rates \cite{Sevcik1971}.
 In this setting,  optimal means optimal with respect to 
other nonanticipatory  policies, and more precisely 
SEPT has the smallest expected total completion time of 
all nonanticipatory  policies for this class of distributions. 
The  Gittins (Index) policy 
is optimal
(again relative to other possible nonanticipatory policies) for all distributions~\cite{gittins1979bandit,Gittens}.
As it is somewhat involved, we postpone the formal definition of
the Gittins policy until Section \ref{sec: prelims},
but roughly, at each time the Gittins policy commits
to running the prefix of a job that maximizes the probability of completing during that prefix relative to the length of that prefix (appropriately accounting for the probability that the job may complete during the prefix).

Some  stochastic scheduling policies are (at least intuitively) robust,
in that they remain nearly optimal if the predicted
distributions are nearly correct. One example of this is the SEPT policy on distributions with increasing hazard rates~\cite{Sevcik1971}.
Intuitively, SEPT is robust because its
design and analysis only use unconditional expected values, which are 
relatively robust with respect to small errors in the reported distribution.
In terms of the analysis, 
if the jobs are ordered by increasing expected processing time, then the expected total completion time
of SEPT is $\sum_{j=0}^{n-1} (n-j) \cdot \mathbb{E}[P_j]$.

In contrast, other stochastic scheduling policies are quite brittle/non-robust, in that even small errors
in the reported distributions can result in quite poor schedules. For many stochastic scheduling problems, the known algorithms require carefully rescaled and truncated statistics based on exponential \cite{kleinberg1997allocating, gupta2021stochastic} and $p$th moments \cite{molinaro2019stochastic}, which are highly sensitive to errors.
As we will observe, the most natural candidate for
a robust policy for the scheduling problem that we consider, namely the Gittins policy itself, is another such example. 

The first step toward addressing the open question 
of whether there is a robust scheduling policy is
to select an appropriate definition of distance between
distributions so that we have a measure of the error in
the reported distributions. For this purpose, we introduce the following error measure.

\begin{definition} \label{def: dist_err}
   Consider two probability distributions, $\dis$ and $\dis'$,
 over $\mathbb{R}_{\geq 0}$. Let $\alpha \geq 1$ be a constant.
    Then  the pair $(\dis, \dis')$ is \emph{$\alpha$-close}  if for all $x \geq 0$, it is the case that
    \[
    \frac{1}{\alpha} \cdot \pr_{P \sim \dis}(P > \alpha x) \leq \pr_{P' \sim \dis'} (P' > x) \leq \alpha \cdot \pr_{P \sim \dis}\left(P > x/\alpha \right).
    \]
\end{definition}

Two distributions are close to each other under this definition if
the relative error in their (approximate) upper tails is always small. It is useful to consider the case that $\alpha = 1 + \varepsilon$ for some small $\varepsilon$,
in which case Definition \ref{def: dist_err} roughly says the relative error between the upper tails is at most $\varepsilon$. 
In Section \ref{sec: prelim_err}, we further motivate our choice of this distance measure
by  establishing some natural properties
that it satisfies (subsection \ref{subsec: measure-properties}), exhibiting some natural examples of distributions that are close in this sense (subsection \ref{subsec: measure-examples}), and comparing this  distance measure with other standard definitions of distance between distributions (subsection \ref{subsec: relation-distances}). Of particular note is that our definition of $\alpha$-close is
symmetric, so switching the roles of  $\dis$ and $\dis'$
does not effect the closeness of these distributions. 

With this definition of closeness in hand, our first contribution is showing that the
Gittins policy is brittle with respect to small
errors in the predicted distributions.
We then run into the issue that there are several
alternative formulations of the Gittins policy,
which, while all equivalent on perfectly accurate predicted distributions,
differ in their executions when there are errors in the distributions.
We elect to consider the formulation of the Gittins Index Priority Policy,
abbreviated as \textsf{GIPP}, that
we think is the most natural
candidate  for application
to noisy data. The one we use is based on the partitioning of jobs into subjobs called quanta (singular: quantum).
We postpone delving into this issue more deeply until Section \ref{sec: prelims}.
Having settled on this version of Gittins, we can
then show that the Gittins policy is brittle, which
requires some definitions. 

\begin{definition} \label{def: nonaticipatory}
Let
$\hat{\mathcal I} = \{\hat{\mathcal D}_j\}_{j=1}^n$
and $\mathcal{I}^* = \{\mathcal D_j^*\}_{j=1}^n$
be  collections of non-negative job size distributions. We let  
$\textsf{A}(\mathcal{I}^*, \hat{\mathcal I})$ be a nonanticipatory policy that is given access to $\hat{\mathcal I}$ at the beginning of time and is
 unaware of $\mathcal{I}^*$. At any fixed time, the policy knows  the completed job sizes and how much it has processed each incomplete job, where the job sizes are drawn from $\mathcal{I}^*$. 
\end{definition}

\begin{definition} \label{def: alg-cost}
Let $\textsf{A}( \cdot, \cdot)$ be a policy as in \Cref{def: nonaticipatory}, and
$\hat{\mathcal I} = \{\hat{\mathcal D}_j\}_{j=1}^n$
and $\mathcal{I}^* = \{\mathcal D_j^*\}_{j=1}^n$
be  collections of job size distributions.
Then we conflate
$\textsf{A}(\mathcal{I}^*, \hat{\mathcal I})$ to represent the policy itself and also the \emph{expected} total completion time for policy
\textsf{A}, using the  instance $\hat{\mathcal I}$ as the predicted distributions 
that are initially provided to $\textsf{A}$, and when 
the instance $\mathcal{I}^*$ is the true distribution on job sizes. 
\end{definition}

Note that while $\textsf{A}(\mathcal{I}^*, \hat{\mathcal I})$ 
may not necessarily be well-defined, it will be well-defined within the context of our use. In particular, with the notation defined above, we can express the optimality of Gittins (\gittins) as follows. We let $\textsf{OPT}(\mathcal{I})$ be the optimal expected cost among all nonanticapatory scheduling policies with input $\mathcal{I}$. 

\begin{restatable}[\cite{Konheim1968, Sevcik1971, Weiss}]{theorem}{propgittinsopt}\label{prop: gittins_opt}
    For any instance $\mathcal{I} = \{\mathcal{D}_j\}_{j \in [n]}$, it is the case that $\gittins(\mathcal{I}, \mathcal{I}) = \opt(\mathcal{I})$.
\end{restatable}

For brevity, we also define $\gittins(\mathcal{I}) = \gittins(\mathcal{I}, \mathcal{I})$ (the latter in the sense of \Cref{def: alg-cost}). We are ready to state our lower bound, which states that \gittins is \emph{not} robust to mis-specified predicted distributions -- even if those predictions are arbitrarily close to the true distributions.

\begin{restatable}{theorem}{thmlower} \label{thm: lower} 
    For all $\alpha > 1$, and for all $n \geq 2$, 
    there exist true distributions $\dis^*_j = \dis^*_j(\alpha, n)$ which depend on $\alpha$ and $n$ for all $j \in [n]$ and predicted distributions, $\hat{\dis}_j = \hat{\dis}_j(n)$ which depend on $n$ for all $j \in [n]$. All such distributions are finitely supported, and \emph{every} pair $(\dis_j^*, \hat{\dis}_j)$ is $\alpha$-close. Then the instances $\mathcal{I}^* = \{\dis_j^*\}_{j \in [n]}$ and $\hat{\mathcal{I}} = \{\hat{\dis}_j\}_{j \in [n]}$ satisfy
    \[\textsf{GIPP}(\mathcal{I}^*, \hat{\mathcal I})  = \Omega(n) \cdot \textsf{GIPP}(\mathcal{I}^*, \mathcal{I}^*).\]
\end{restatable}

In other words, \Cref{thm: lower} says that that the cost of the Gittins policy using erroneous predicted distributions (which are even arbitrarily close to the true distributions) grows at an unbounded (in fact, linear in the number of jobs) rate compared to the optimal cost. Thus, to obtain any near-optimal policy with erroneous distributions, we need a new algorithm.

Intuitively, the  reason for the brittleness of the Gittins policy is that both the design and analysis of Gittins depend on 
conditional probability distributions derived from the job size distributions, 
and for these instances the conditional probability distributions can be quite brittle with respect to small errors in the predicted job
size distributions. 
In particular, in the proof of Theorem \ref{thm: lower},  in each true distribution $\mathcal{D}^*_j $ the
size $P_j^* \sim \dis_j^*$ is likely $1+\varepsilon$, but with probability $1/M$ it
is of some big size $M$, and in 
each reported distribution $\hat{\mathcal{D}_j }$ the
size $\hat{P}_j \sim \hat{\dis}_j$ is likely $1$, but with probability $1/M$ it
is of some large size $M$. 
Using these predicted distributions,
the Gittins policy will deprioritize jobs once
they have run for a unit of time because 
it thinks that their conditional expected remaining processing
time is large, when in fact it is small. 

We then turn to the open question of whether
there is a robust policy for this problem.
Our main contribution is a positive answer to this question.
Our robust policy is a modest modification of the
Gittins policy that  naturally arises from consideration
of the lower bound instances  in the 
proof of Theorem \ref{thm: lower}. 
We uncreatively call this new policy the Robust Gittins 
policy, or more succinctly, \textsf{RG}. Roughly, if the Gittins policy would
preempt a job $j$ after running it continuously for $q$
time units, the Robust Gittins policy would run
$j$ for an additional $(\alpha - 1) q$ time units. 
Thus in the lower bound instance for Theorem \ref{thm: lower},
the Robust Gittins policy would not make the mistake that
the Gittins policy makes of preempting a job right
before it is done. We are now ready to state our
bound on the performance of the Robust Gittins policy. Note that our result assumes true and predicted distributions have finite support; see subsection \ref{subsec: gittins-finite-support} for further discussion on this assumption.

\begin{restatable}{theorem}{thmupper} \label{thm:clairvoyant_main} 
Let $\alpha \ge 1$. Let
$\mathcal{I}^* = \{\dis_j^*\}_{j \in [n]}$ be a collection of true size distributions with finite support on $n$ jobs, and $\hat{\mathcal{I}} = \{\hat{\dis}_j\}_{j \in [n]}$ be a collection of predicted size distributions with finite support on $n$ jobs. Further assume that \emph{every} pair of distributions $(\dis^*_j, \hat{\dis}_j)$ is $\alpha$-close.
Then
\[\textsf{RG}(\mathcal{I}^*, \hat{\mathcal I})  \le  \alpha^6 \cdot\textsf{GIPP}(\mathcal{I}^*, \mathcal{I}^*).\]
\end{restatable}

Note that the right-hand side is just the optimal expected cost for the \textit{true} distributions. So in the case that $\alpha =  1+ \varepsilon$ for some
small $\varepsilon$,  Theorem \ref{thm:clairvoyant_main}
states that the expected total completion time for our Robust Gittins policy 
is roughly within a $(1+ 6 \varepsilon)$ factor of the optimal expected total completion time,
 despite receiving the potentially erroneous predicted job size distributions as input. 

Theorem 
\ref{thm:clairvoyant_main} 
is proven in Section \ref{sec:RG} using the following sequence of inequalities:
\[\rg(\mathcal{I}^*, \hat{\mathcal{I}}) \leq \alpha^3 \cdot \gittins(\hat{\mathcal{I}}, \hat{\mathcal{I}}) \leq \alpha^3 \cdot \rg(\hat{\mathcal{I}}, \mathcal{I}^*) \leq \alpha^6 \cdot \gittins(\mathcal{I}^*, \mathcal{I}^*). \]
These inequalities are established using the fact that $\mathcal{I}^*$
and $\hat{\mathcal{I}}$ are $\alpha$-close, the similarity of 
the $\rg$ policy and the $\gittins$ policy, and
the optimality of the $\gittins$ policy. But of particular note is
that, as one can see from these inequalities, our analysis relies on the symmetry of the definition of $\alpha$-close, and this symmetry is to a significant extent responsible for
the relative cleanliness of our analysis.

\subsection{Related Work} \label{sec: related-work}

  The literature on stochastic scheduling is sufficiently large that covering it here would be overly
  ambitious. A good starting point into the literature is the text by Pinedo~\cite{Pinedo2022}. 
  We will content ourselves by mentioning a few additional results for the particular scheduling problem
  that we consider.  Chazan \cite{Chazan1968} and Konheim \cite{Konheim1968} gave necessary and sufficient conditions for a policy to be optimal. Sevcik \cite{Sevcik1971} introduced an intuitive method to create an optimal policy and later Weiss \cite{Weiss} showed the connection to Gittins index. 
  If the job sizes are deterministic, then the online policy  Smallest Processing Time (SPT),
which always runs the job of smallest size, is optimal. 

Our model for analyzing the effect of inaccurate predicted distributions is most similar to that of D{\"u}tting and Kesselheim \cite{prophet-inaccurate-priors} on prophet inequalities with inaccurate priors, and of Banihashem et. al. \cite{banihashem2025pandora} on Pandora's box with inaccurate priors. Both papers use standard distance metrics to measure the distance between inaccurate and accurate prior distributions, and use this distance to parameterize the error in welfare / utility that results from using existing algorithms but with inaccurate priors. We follow the same framework, but key differences are that we use a novel error metric more suited to the problem, and, due to strong lower bounds, develop a novel algorithm.
  
 The literature on stochastic scheduling, on the other hand, that considers the effect of mispredictions in the
 reported distributions is comparatively small. We now give an overview of the  papers in the literature
 that are closest to this work. We emphasize that these results concern the M/G/1 queue, which is a different scheduling environment from ours; further, one is not a special case of the other. In the M/G/1 queue, jobs arrive over time according to a Poisson process, but all jobs are identically distributed. On the other hand, in our problem, all jobs are available at time $0$, but they can have arbitrary (non-identical) independent distributions.

\paragraph{Stochastic Scheduling with Mispredictions:} Scully, Grosof, and Mitzenmacher \cite{scully_et_al:LIPIcs.ITCS.2022.114} considered erroneous job size information in the M/G/1 queue. They give the scheduler a stochastically drawn estimate $z_j$ of the true (stochastically drawn) job size $s_j$ of job $j$, rather than an estimate $\hat{\mathcal{D}}$ of the true job size distribution $\mathcal{D}$ as in our setting. The $(s_j, z_j)$ are drawn i.i.d. from a single joint distribution $(S,Z)$; the support of $(S,Z)$ is restricted so that the estimated size has bounded \textit{multiplicative} error from the true size: $z_j \in [\beta s_j, \alpha s_j]$ for some $\alpha > \beta \geq 0$, $\beta < 1$. Because the scheduler is given the realization $z_j$, rather than distributional information as in our setting, the authors compare their scheduling policy's cost to that of Shortest Remaining Processing Time (SRPT) (which is optimal for \textit{known} job sizes \cite{schrage1968proofSRPT}) rather than to the Gittins policy, which turns out to also be optimal in the M/G/1 setting with \textit{unknown} job sizes \cite{scullyGittinsOpt}. The authors achieve graceful degradation of performance with respect to error using a variant of SRPT. They also contribute a result bounding the multiplicative gap, dubbed the \textit{price of misprediction} in \cite{mitzenmacher-price-misprediction}, between naively running the Shortest Job First policy using the job size estimates and running the same (suboptimal) policy using the true sizes. Wierman and Nuyens \cite{wierman} also considered SRPT and variants in the M/G/1 queue, but use an \textit{additive} notion of error that they additionally assume updates over time.

Scully and Harchol-Balter \cite{scully-age-error} likewise work in the M/G/1 queue. They considered error not in the scheduler's knowledge of job size distributions, but rather, in their knowledge of the ages (amount of service a job has received so far) of jobs at any point in time, which is critical information for the Gittins policy. They are motivated by examples in time-shared or networked systems where acquiring / updating age information incurs cost. Unlike in \cite{scully_et_al:LIPIcs.ITCS.2022.114}, they consider \textit{additive} (vs. multiplicative), \textit{adversarially chosen} (vs. stochastic) error: at any point in time, the scheduler knows a perturbed age $b \in [a-\Delta, a+\Delta]$ for a job with exact age $a$, where $\Delta \geq 0$ is the (maximum amount of) error. As in our setting, the scheduler has knowledge of the error bound $\Delta$. The authors first show that the naive Gittins policy, which uses the perturbed ages as given in order to compute the rank function (see Section \ref{sec: prelims}), is \textit{not} a robust policy for mean response time. (The analogous result in our setting is \Cref{thm: lower}.) Then, they show that their \textit{shift-flat Gittins policy}, a variation of Gittins that shifts a job's perturbed age back by $\Delta$ before computing rank, then flattens the rank function around local maxima, is robust, meaning that as $\Delta \to 0$, the expected cost of the policy approaches the expected cost of the Gittins policy under zero noise. We note that, despite the difference between our setting and theirs in terms of which information is erroneous, our scheduling policy is similar in spirit to the shift-flat policy. 

Note that in addition to the above works considering the M/G/1 queue, another difference from our setting is the model of error. On the other hand, there is also a line of work on scheduling with predictions when job sizes are unknown in the non-stochastic setting, where both release dates (when they exist) and job sizes are deterministic \cite{purohit2018improving, im2023non,  azar2021flow, azar2022distortion, lindermayr2022permutation}.

\paragraph{Multi-armed Bandits:} While not set in a stochastic scheduling setting, we also mention the work of Kim and Lim \cite{kim2016robust} on making the Gittins index robust to distributional errors in the setting of multi-armed bandits, which recall is the setting in which the Gittins index policy was originally formulated \cite{Gittens, gittins1979bandit}. They model uncertainty in the transition probability distributions as follows: in response to each choice of arm by the decision-maker, nature chooses a transition probability distribution adversarially. In their ``robust bandit problem,'' confidence in the given probability distribution model is expressed by building into the value function a cost to nature's deviations from the given model, which is quantified using relative entropy. The authors define a (suboptimal) ``robust Gittins index'' policy that is more tractable to compute than the optimal policy for the robust bandit problem. They also provide experiments in which their robust Gittins index outperforms naive Gittins for the robust bandit problem.

\section{New Error Measure for Distributions} \label{sec: prelim_err}

The goal of this section is to motivate our choice of distance measure between distributions
by  establishing some natural properties
that it satisfies, exhibiting some natural examples of   distributions that are close in this sense, and comparing this  distance measure with other standard definitions of distance between distributions.
 We first generalize the  definition of $\alpha$-close distributions to \emph{families} of distributions by applying the above definition to each pair in the family.

\begin{definition} \label{def: dist_err_family}
    Given families of probability distributions $\{\dis_j\}_{j \in [n]}$ and $\{\dis'_j\}_{j \in [n]}$, all over $\mathbb{R}_{\geq 0}$, we say the pair of families $(\{\dis_j\}_{j \in [n]}, \{\dis'_j\}_{j \in [n]})$ is \emph{$\alpha$-close} for $\alpha \geq 1$ if for all $j \in [n]$, the pair of distributions $(\dis_j, \dis_j')$ is $\alpha$-close in the sense of \Cref{def: dist_err}.
\end{definition}

Next, we will show that our error measure satisfies some natural properties.

\subsection{Basic Properties} \label{subsec: measure-properties}

As stated, our error measure is defined for an \emph{ordered pair} of distributions, $(\dis, \dis')$. Our first property is that actually this ordering is irrelevant.

\begin{lemma}[Symmetric]\label{lem: dist_sym}
    If distribution pair $(\dis, \dis')$ is $\alpha$-close in the sense of \Cref{def: dist_err}, then $(\dis', \dis)$ is also $\alpha$-close.
\end{lemma}
\begin{proof}
    Suppose $(\dis, \dis')$ is $\alpha$-close. We let $P \sim \dis$ and $P' \sim \dis'$. Then for any $x' \geq 0$, the definition of $\alpha$-close applied to $(\dis, \dis')$ and $x = \alpha x'$ gives
    \[\pr(P' > \alpha x') \leq \alpha \cdot \pr(P > x') \Rightarrow \frac{1}{\alpha} \cdot \pr(P' > \alpha x') \leq \pr(P > x'),\]
    via the right-side inequality. Similarly, taking $x = \frac{x'}{\alpha}$ gives
    \[\frac{1}{\alpha} \cdot \pr(P > x') \leq \pr(P' > x'/\alpha) \Rightarrow \pr(P > x') \leq \alpha \cdot \pr(P' > x'/\alpha)
    ,\]
    via the left-side inequality. Combining the above two bounds on $\pr(P > x')$ for all $x' \geq 0$ gives that $(\dis', \dis)$ is $\alpha$-close, as required.
\end{proof}

In light of the above lemma, instead of saying that ``$(\dis,\dis')$ is $\alpha$-close", from now on we will simply say ``$\dis$ and $\dis'$ (or equivalently, $\dis'$ and $\dis$) are $\alpha$-close." Similarly, when extending \Cref{def: dist_err} to families of distributions $\{\dis_j\}_{j \in [n]}$ and $\{\dis'_j\}_{j \in [n]}$, we will now say that the families ``$\{\dis_j\}_{j \in [n]}$ and $\{\dis'_j\}_{j \in [n]}$ are $\alpha$-close" rather than specifying an order. Next, we show that our error measure is monotone in $\alpha$.

\begin{lemma}[Monotone]\label{lem: dist_mono}
    Let $1 \leq \alpha < \alpha'$. If $\dis$ and $\dis'$ are $\alpha$-close, then they are also $\alpha'$-close. Further, $\dis$ and $\dis'$ are $1$-close if and only if $\dis$ and $\dis'$ are identical distributions.
\end{lemma}
\begin{proof}
    Consider any $1 \leq \alpha < \alpha'$, and suppose $\dis$ and $\dis'$ are $\alpha$-close. We let $P \sim \dis$ and $P' \sim \dis'$. Then for any $x \geq 0$ we have
    \[
        \frac{1}{\alpha'} \cdot \pr(P > \alpha' x) \leq \frac{1}{\alpha} \cdot \pr(P > \alpha x),
    \]
    and
    \[
        \alpha \cdot \pr\left(P > x/\alpha\right) \leq \alpha' \cdot \pr\left(P > x/\alpha'\right).
    \]
    Combining the above two inequalities with the definition of $\alpha$-close gives that $\dis$ and $\dis'$ are also $\alpha'$-close.

    Next, suppose $\dis$ and $\dis'$ are $1$-close. Then for any $x \geq 0$, we have $\pr(P > x) = \pr(P' > x)$. Since $P$ and $P'$ are non-negative, this means that the CDFs of $P$ and $P'$ are equal, i.e., that $\dis$ and $\dis'$ are the same distribution.
\end{proof}

Finally, our error measure satisfies a natural (multiplicative) ``triangle inequality." This is particularly useful to compose different operations which lead to $\alpha$-close distributions (some of which we introduce in the next section).

\begin{lemma}[Composition]\label{lem: dist_comp}
    Let $\dis$, $\dis'$, $\dis''$ be distributions such that $\dis$ and $\dis'$ are $\alpha_1$-close and $\dis'$ and $\dis''$ are $\alpha_2$-close. Then $\dis$ and $\dis''$ are $\alpha_1 \alpha_2$-close.
\end{lemma}
\begin{proof}
    We let $P \sim \dis$, $P' \sim \dis'$, and $P'' \sim \dis''$. Now consider any $x \geq 0$. Since $\dis$ and $\dis'$ are $\alpha_1$-close, we have
    \[
    \frac{1}{\alpha_1} \cdot \pr(P' > \alpha_1 x) \leq \pr(P > x) \leq \alpha_1 \cdot \pr\left(P' > x/\alpha_1 \right).
    \]
    Then using the $\alpha_2$-closeness of $\dis'$ and $\dis''$ gives
    \[\frac{1}{\alpha_2} \cdot \pr (P'' > \alpha_2 \cdot \alpha_1 x) \leq \pr(P' > \alpha_1 x),\]
    and
    \[\pr\left(P' > x/\alpha_1\right) \leq \alpha_2 \cdot \pr(P'' > x/(\alpha_1 \alpha_2)).\]
    Combining the above inequalities gives that $\dis$ and $\dis''$ are $\alpha_1 \alpha_2$-close, as required.
\end{proof}

\subsection{Examples} \label{subsec: measure-examples}

Now that we have established that our error measure has some natural and desirable properties, we  give a few illustrative examples of distributions that are $\alpha$-close. Our first example is actually the initial motivation for our error measure, which is to capture the error in estimating a discrete distribution through its support points and their respective probabilities.

\begin{example}[Discrete distributions]\label{ex: discrete}
    Consider any non-negative discrete distribution $\dis$. We can specify $\dis$ in terms of its atoms $\{(s_i, p_i)\}_{i \in I}$, where $I$ is some countable index set and $\pr_{P \sim \dis}(P = s_i) = p_i$ for all $i$ such that $\sum_{i \in I} p_i = 1$.

    We now introduce three operations that modify a discrete distribution such that the resulting distribution is $\alpha$-close. For the first two operations, it is convenient to think of a discrete distribution as a histogram, and these operations as shifting the vertical/horizontal positions of each bar. The third operation simply combines the first two into a single one for convenience of analysis. We say each of the shifts have parameter $\alpha$.

    \begin{itemize}
        \item \textbf{Vertical Shift:} We modify $\dis$ to obtain $\dis'$ by modifying the probabilities of the atoms: $\dis$ and $\dis'$ are specified by $\{(s_i, p_i)\}_{i \in I}$ and $\{(s_i, p'_i)\}_{i \in I}$, respectively, where for all $i \in I$, $p'_i \in [\frac{1}{\alpha} \cdot p_i, \alpha \cdot p_i]$ (and of course maintaining that $\sum_{i \in I} p_i' = 1$).
        \item \textbf{Horizontal Shift:} We modify $\dis$ to obtain $\dis'$ by modifying the support points of the atoms. For each support point $s_i$ for $i \in I$, we can replace $s_i$ with any discrete collection of support points $\{s_{i,k}\}_k$ with corresponding probabilities $\{p_{i,k}\}_k$ such that $s_{i,k} \in [\frac{1}{\alpha} \cdot s_i, \alpha \cdot s_i]$ for all $k$ and $\sum_k p_{i,k} = p_i$. Then $\dis'$ is defined by the atoms $\bigcup_{i \in I} \{(s_{i,k}, p_{i,k})\}_k$.
        \item \textbf{Combined Shift}: We modify $\dis$ to obtain $\dis'$ as follows. For each support point $s_i$ for $i \in I$, we can replace $s_i$ with any discrete collection of support points $\{s_{i,k}\}_k$ with corresponding probabilities $\{p_{i,k}\}$  such that $s_{i,k} \in [\frac{1}{\alpha} \cdot s_i, \alpha \cdot s_i]$ for all $k$ and $\sum_k p_{i,k} \in [\frac{1}{\alpha} \cdot p_i, \alpha \cdot p_i]$. Then $\dis'$ is defined by the atoms $\bigcup_{i \in I} \{(s_{i,k}, p_{i,k})\}_k$. 
    \end{itemize}
\end{example}

    We first verify that all three operations lead to $\alpha$-close distribution pairs. Note that combined shifts generalize both vertical and horizontal shifts, so we need only consider combined shifts in the next proposition. 
    
    \begin{proposition}\label{prop: shift}
        Let $\dis$ be any non-negative discrete distribution and $\dis'$ be obtained from $\dis$ by a combined shift with parameter $\alpha$. Then $\dis$ and $\dis'$ are $\alpha$-close.
    \end{proposition}
    \begin{proof}
        Let $\dis$ be specified by atoms $\{(s_i, p_i)\}_{i \in I}$ and $\dis'$ by $\bigcup_{i \in I} \{(s_{i,k}, p_{i,k})\}_k$. Then let $P \sim \dis$ and $P' \sim \dis'$. For any $x \geq 0$, we have $\pr(P' > x) = \sum_{i \in I} \sum_{k : s_{i,k} > x} p_{i,k}$. Note that by definition of combined shift, $s_{i,k} > x$ implies $s_i > \frac{1}{\alpha} x$, and for any $i$ we have $\sum_k p_{i,k} \leq \alpha p_i$. Then we have
        \[\pr(P' > x) \leq \sum_{i \in I : s_i > \frac{1}{\alpha} x} \alpha p_i = \alpha \cdot \pr(P > x/\alpha).\]
        Analogously, we have $\pr(P > \alpha x) = \sum_{i \in I : s_i > \alpha x} p_i$. Again, if $s_i > \alpha x$, then for any $k$ we have $s_{i,k} > x$, so
        \[\pr(P > \alpha x) = \sum_{i \in I : s_i > \alpha x} p_i \leq \alpha \cdot \sum_{i \in I : s_i > \alpha x} \sum_k p_{i,k} \leq \alpha \cdot \sum_{i \in I} \sum_{k : s_{i,k} > x} p_{i,k} = \alpha \cdot \pr(P' > x).\]
        Combining the two equations gives that $\dis$ and $\dis'$ are $\alpha$-close, as required. 
    \end{proof}

    One might wonder if a sequence of combined shifts completely characterize $\alpha$-close discrete distributions. We give an example to show that this is \textit{not} the case.

    \begin{proposition}\label{prop: shift_ex}
        For any $\alpha > 1$, there exist non-negative discrete distributions $\dis$ and $\dis'$ that are $\alpha$-close to each other, but one cannot be obtained from the other by a combined shift.
    \end{proposition}
    \begin{proof}
        Consider any fixed $\alpha > 1$. Then we define the distributions $P \sim \dis$ and $ P'  \sim \dis'$ such that $P = 1$ and 
        \[P' =  \begin{cases}
        \frac{1}{\alpha^2} & w.p. \quad 1 - \frac{1}{\alpha}\\
        1 & w.p. \quad \frac{1}{\alpha}
    \end{cases}.\] 
    We first check that $\dis$ and $\dis'$ are $\alpha$-close. We observe that $P' \leq P$, which implies for any $x \geq 0$ that $\pr(P' > \alpha x) \leq \pr(P' > x) \leq \pr(P > x)$.
    
    It remains to show that $\pr(P > x) \leq \alpha \cdot \pr(P' > \frac{1}{\alpha} x)$ for all $x \geq 0$. This inequality is trivially true for all $x \geq 1$, since in this case $\pr(P > x) = 0$. Otherwise, for $x \in (0,1)$, we have $\pr(P > x) = 1$, so it remains to show that for all such $x \in (0,1)$, we have $\frac{1}{\alpha} \leq \pr(P' > \frac{1}{\alpha} x)$. One can check that
    \[\pr(P' > x/\alpha) = \begin{cases}
        1 ,& x < \frac{1}{\alpha}\\
        \frac{1}{\alpha} ,& \frac{1}{\alpha} \leq x < 1
    \end{cases},\]
    which gives the desired inequality. We conclude that $\dis$ and $\dis'$ are $\alpha$-close.

    Second, we show that $\dis'$ \textit{cannot} be obtained from $\dis$ by a combined shift (and vice-versa). To see this, we observe that the support point $\frac{1}{\alpha^2}$ of $\dis'$ is \textit{not} within a $\alpha$-factor of $1$, the only support point of $\dis$. 
    \end{proof}

Note that our error measure applies to any non-negative distribution, not just discrete ones. Our next two examples show that for parameterized families of distributions, we can translate errors in the parameter into bounds in our error measure.

Our first example is the exponential distribution, which is parametrized by $\lambda > 0$ such that $P \sim Exp(\lambda)$ satisfies $\pr(P > x) = e^{- \lambda x}$ for all $x \geq 0$.

\begin{proposition}[Exponential distribution] \label{prop exp_close}
    Let $0 < \lambda \leq \lambda'$. Then the distributions $Exp(\lambda)$ and $Exp(\lambda')$ are $\frac{\lambda'}{\lambda}$-close.
\end{proposition} 
\begin{proof}
    We let $P \sim Exp(\lambda)$ and $P' \sim Exp(\lambda')$, where $\lambda \leq \lambda'$. Then by by definition of the exponential distribution, for any $x \geq 0$, we have
    \[\pr\left(P' > \frac{\lambda'}{\lambda} x\right) \leq \pr(P' > x) = e^{-\lambda' x} \leq e^{-\lambda x} = \pr(P > x),\]
    and
    \[\pr(P > x) = e^{-\lambda x} = e^{-  \lambda' \frac{\lambda}{\lambda'} x} = \pr\left(P' > \frac{\lambda}{\lambda'} x\right).\]
    Combining the above two inequalities gives that $Exp(\lambda)$ and $Exp(\lambda')$ are $\frac{\lambda'}{\lambda}$-close, as required. 
\end{proof}

Our second example is the Pareto distribution, which is parametrized by $m > 0$ and $\beta >0$ such that $P \sim Pareto(m,\beta)$ satisfies $\pr(P > x) = \begin{cases}
\left(\frac{m}{x}\right)^\beta & x \geq m\\
1 & x < m
\end{cases}$.

\begin{proposition}[Pareto distribution] \label{prop pareto_close}
    Let $0 < m \leq m'$ and $\beta > 0$. We define $\alpha = \frac{m'}{m}$ Then the distributions $Pareto(m, \beta)$ and $Pareto(m', \beta)$ are $\frac{m'}{m}$-close.
\end{proposition}
\begin{proof}
    We let $P \sim Pareto(m, \beta)$ and $P' \sim Pareto(m', \beta)$, where $m \leq m'$. For any $x \geq 0$, we split into a few cases.

    First if $x < m'$, we have $\pr(P' > x) = 1$, so trivially $\frac{1}{\alpha} \cdot \pr(P > \alpha x) \leq 1$. For the other inequality, we observe that $\frac{x}{\alpha} = \frac{m}{m'} \cdot x < m$, so $\alpha \cdot \pr(P > \frac{x}{\alpha}) = \alpha \geq 1$. Combining these bounds gives the desired inequality for the case $x < m'$.

    Otherwise, $x \geq m'$. In this case $\pr(P' > x) = \left( \frac{m'}{x} \right)^\beta$. We have $\pr(P > \alpha x) = \left( \frac{m}{\alpha x} \right)^\beta = \alpha^{- \beta} \left( \frac{m'}{x} \cdot \frac{m}{m'} \right)^\beta = \alpha^{-2\beta} \left( \frac{m'}{x} \right)^\beta \leq \pr(P > x)$. Similarly, $\pr(P > x / \alpha) = \left( \frac{\alpha m}{x} \right)^\beta = \left( \frac{m'}{x} \right)^\beta = \pr(P' > x)$. We conclude that $Pareto(m, \beta)$ and $Pareto(m', \beta)$ are $\alpha = \frac{m'}{m}$-close, as required.
\end{proof}

We note that deviations in the $\beta$-parameter lead to distributions which are \emph{not} $\alpha$-close for any constant $\alpha$. To see this, consider $P \sim Pareto(1, \beta)$ and $P' \sim Pareto(1, \beta')$ with $0 < \beta < \beta'$. Now fix any $\alpha \geq 1$ and consider $x > 1$. Then $\pr(P > \alpha x) = (\alpha x)^{- \beta} = \alpha^{-\beta} \cdot x^{- \beta'} x^{\beta' - \beta} = \alpha^{-\beta} x^{\beta' - \beta} \cdot \pr(P' > x) = \omega(\pr(P' > x))$.

\subsection{Relation to Other Distances} \label{subsec: relation-distances}

There are a number of known distance measures on probability distributions; notably, all standard ones of which we are aware are \textit{additive}. We refer the reader to Section 3 of \cite{prophet-inaccurate-priors} for an excellent discussion of the standard ones and the relationships between them.

The distance measure most related to our $\alpha$-closeness definition (Definition \ref{def: dist_err}) is the L{\'e}vy metric. Given two distributions $\mathcal{D}, \mathcal{D}'$ with cumulative distribution functions $F$ and $F'$, respectively, the L{\'e}vy distance between them is defined as 
\[d_L(\mathcal{D}, \mathcal{D}') = \inf\{\varepsilon \geq 0 \mid F'(x-\varepsilon) - \varepsilon \leq F(x) \leq   F'(x+\varepsilon) + \varepsilon  \quad \forall x \in \mathbb{R}\}.\]
Concretely, one can visualize distributions that are $\varepsilon$-close according to the L{\'e}vy distance as follows: form two bands by taking the CDF $F'$ of $\mathcal{D}'$ and shifting it to the left by $\varepsilon$ and up by $\varepsilon$, and again but to the right by $\varepsilon$ and down by $\varepsilon$. Then the CDF $F$ of $\mathcal{D}$ must lie between these two bands. In this way, we can think of the L{\'e}vy metric as allowing both horizontal and vertical shifts.\footnote{Other distance measures, like the Wasserstein metric (Earth Mover's Distance) $d_W$, also capture this notion, and in fact $d_W \leq 4d_L$ and $d_L \leq \sqrt{d_W}$ for distributions supported on $[0,1]$ \cite{prophet-inaccurate-priors}. It is unclear to us whether there is a multiplicative analogue of Earth Mover's Distance with a similar relationship to our distance measure.}

Our definition of $\alpha$-closeness can be thought of as a multiplicative version of the L{\'e}vy metric, as it likewise allows for vertical and horizontal shifts. (Note that a stronger requirement is given by the Kolmogorov metric, which does not allow for horizontal shifts as in the L{\'e}vy metric but is otherwise the same. We likewise could have defined our error measure to only allow vertical shifts, but we obtain strictly more general results by also accommodating horizontal shifts.)

The use of a multiplicative distance measure is for practical and technical reasons. Intuitively, this is because we are willing to tolerate larger errors for longer jobs (a 1 minute error in a job that takes 1 hour is much less impactful than if the job only took 1 second). Technically, suppose the predicted distributions $\{\hat{\mathcal{D}}_j\}_{j \in [n]}$ are each deterministic with support 1. To form the $n$ true distributions $\mathcal{D}^*_j$, add a single support point with arbitrarily large value $M_j$ and small probability $\varepsilon$, and let the $M_j$ increase arbitrarily with $j$. Then $d_L(\mathcal{D}^*_j, \hat{\mathcal{D}}_j) \leq \varepsilon$ for all $j$. An algorithm receiving the $n$ predicted distributions as input is effectively a nonclairvoyant algorithm once all jobs complete 1 unit of work. The knowledge that the predicted distributions are $\varepsilon$ distance away from the true distributions in the L{\'e}vy metric is irrelevant for determining in what order to run the jobs that realize to sizes larger than 1, so the algorithm must now act clairvoyantly.

\section{Gittins Preliminaries} \label{sec: prelims}

We first introduce some notations and concepts, which we will need to describe the Gittins index policy and, later, our modification of it. Let $\mathcal{I} = \{\mathcal{D}_j\}_{j \in [n]}$ be an instance of our scheduling problem. We roughly follow the definitions in \cite{MegowV14}, but with some departures in notation.

It is convenient to specify preemptive schedules in terms of \emph{quanta} (singular: quantum). A \emph{quantum} is a scheduling decision, which is defined by a pair $(j, q)$, where $j \in [n]$ indexes a job and $q \geq 0$ is the \emph{length} of the quantum. Running quantum $(j,q)$ corresponds to running job $j$ for $q$ further time units or until job $j$ completes --- whichever comes first. Precisely, if we have already processed job $j$ for $y$ time units, then running quantum $(j,q)$ takes $\min\{P_j - y, q\}$ time units.

There are two main relevant quantities for a quantum $(j,q)$: the \emph{investment} and \emph{rank}. The investment is the expected processing time of quantum $(j,q)$ given that $j$ has already been processed for $y \geq 0$ time units:
\begin{equation} \label{eq: investment-fn}
  I_j(q,y) := \mathbb{E}[\min\{P_j - y, q\} \mid P_j > y].  
\end{equation}

Similarly, the rank of $(j,q)$ given prior processing time $y$ is the ratio of the probability that $j$ completes during this quantum over the investment:
$$R_j(q,y) := \frac{\mathbb{P}(P_j - y \leq q \mid P_j > y)}{I_j(q,y)}.$$
Intuitively, the rank captures the value-per-unit-time of quantum $(j,q)$. Finally, for a given scheduling policy, we let $y_j(t)$ denote the attained processing time of job $j$ at time $t \geq 0$. With these definitions, we are ready to state the Gittins index policy.

\begin{algorithm}[Gittins Index Priority Policy (\gittins)]\label{alg: gittins}
   For each unit of time $t \geq 0$, 
     run the job $j$  achieving $\arg\max_{j \in [n], q \geq 0} R_j \big( q, y_{j}(t) \big)$. In words, based on the attained processing time of each job, run the highest-ranked job over all jobs and quantum lengths.
\end{algorithm}

Note that the policy is well-defined: a maximum is attained because all distributions are assumed to have finite support. 

In the preceding description of \textsf{GIPP}, a scheduling decision is updated at \textit{each} unit of time by recomputing the rank function. It turns out that these updates are not necessary, as the rank function exhibits useful behavior that can be used to simplify the description of the policy. In particular, each scheduling decision will instead be to run a quantum of a job without interruption, and to only update our scheduling decision (of which job to run, and for how long) at the end of running the quantum. There are three properties of this reformulation that, we will see, make it particularly useful in our setting of erroneous predicted distributions. 
\begin{itemize}
    \item The quantum lengths for each job are computed without knowledge of the other job distributions. 
    \item The quantum lengths for each job are precomputed only using the input information $\mathcal{I}$ (i.e., the lengths are not adaptive to realizations of the jobs). 
    \item The order in which quanta are run by \textsf{GIPP} is fixed using only the input information $\mathcal{I}$ (i.e., the order is not adaptive to realizations of the jobs). 
\end{itemize}

We now describe this (known) reformulation that has the above three properties. This reformulation appears in \cite{MegowV14}, and our exposition closely follows theirs. For each job $j$, we compute an ordered set of $n_j$ quanta: $(j, q_{j,1}), \dots, (j, q_{j,n_j})$, which correspond to subjobs that we will run in that order (with interruptions by subjobs from other jobs). The total length of all quanta for job $j$ is equal to the maximum support point of $\mathcal{D}_j$. For each $i \in \{1, \dots, n_j\}$, $q_{j,i}$ is computed recursively as follows. Take $y_{j,1} = 0$. We define $y_{j,i} = q_{j,1} + \cdots + q_{j, i-1}$, that is, $y_{j,i}$ is the attained processing time after the first $i-1$ quanta have been run. Then, define 
\begin{equation} \label{eq: quanta-length}
    q_{j,i} := {\arg \max}_{q \geq 0} R_j(q, y_{j,i}).
\end{equation}

Each quantum $(j, q_{j,i})$ has an associated rank: $R_j(q_{j,i}, y_{j,i})$. Thus, we may order $ \{(j, q_{j,i}) \}_{j \in [n], i \in [n_j]}$, the set of all quanta across all jobs, in decreasing order of rank, breaking ties in favor of lower indexed jobs. We call this ordering the \emph{\gittins order}.

For each job $j$ and $i \in [n_j]$, we define $H(j,i)$ to be the multi-set of all quanta (from all jobs, including job $j$) that occur \emph{no later than} quantum $(j, q_{j,i})$ in the aforementioned \gittins order. Note that for fixed $j$, these multi-sets are nested, that is, $H(j,i) \subseteq H(j, i+1)$. So, we may disjointify these multi-sets for fixed $j$ by defining 
\begin{equation} \label{eq: disjointified-quanta-prefixes}
   H'(j,i) := H(j,i) \setminus H(j,i-1), 
\end{equation}
taking $H(j,0) = \emptyset$. Note that $H'(j,i) \neq \emptyset$ for $i \geq 1$, in particular, that $(j, q_{j,i}) \in H'(j,i)$.

\begin{algorithm}[Gittins Index Priority Policy (\gittins), reformulated]\label{alg: gittins-reformulated}
Schedule job quanta in decreasing order of rank. Run each quantum for its full length or until the job finishes, whichever comes first. 
\end{algorithm}

The fact that Algorithms \ref{alg: gittins} and \ref{alg: gittins-reformulated} are equivalent policies follows from the following fact from \cite{Weiss, Gittens}: for any job $j$ and attained processing time $t$, with $q^* = \arg\max_{q \geq 0} R_j \big( q, y_{j}(t) \big)$, we have that, for any $0 < \zeta < q^*$, 
\[\max_{q \geq 0} R_j \big( q, y_{j}(t) + \zeta \big) \geq \max_{q \geq 0} R_j \big( q, y_{j}(t) \big).\]
For more details on the reformulation of Algorithm \ref{alg: gittins} as Algorithm \ref{alg: gittins-reformulated}, see \cite{MegowV14, aalto2009gittins}. 

From now on, when we refer to \textsf{GIPP}, we will be using the formulation in Algorithm \ref{alg: gittins-reformulated}. We recall that \gittins is optimal, and that we let $\textsf{GIPP}(\mathcal{I})$ refer to the expected cost (total completion time) of \textsf{GIPP} on instance $\mathcal{I}$, that is, $\textsf{GIPP}(\mathcal{I}):=\sum_{j \in [n]} \mathbb{E}[C_j]$, where $C_j$ is the (random) completion time of job $j$ under \textsf{GIPP}. (In the notation of Definition \ref{def: alg-cost}, $\textsf{GIPP}(\mathcal{I}) = \textsf{GIPP}(\mathcal{I}, \mathcal{I})$.)

\propgittinsopt*

Moreover, there is a convenient closed-form expression for the expected cost of \textsf{GIPP} which will be useful for our later analysis. 

\begin{lemma}[Lemma 2.1 in \cite{MegowV14}] \label{lem: closed-form-gipp}
    The expected cost of $\textsf{GIPP}$ on instance $\mathcal{I} = \{\dis_j\}_{j \in [n]}$, where $P_j \sim \dis_j$ for each $j \in [n]$, is given by:

    \begin{align*}
        \gittins(\mathcal{I}, \mathcal{I}) &= \sum_{j=1}^n \sum_{i=1}^{n_j} \sum_{(k,q_{k,l}) \in H'(j,i)} \mathbb{E}\left[\mathbf{1}_{\{P_j > y_{j,i}\}} \cdot \mathbf{1}_{\{P_k > y_{k,l}\}} \cdot \min\{P_k - y_{k,l}, q_{k,l} \}\right] \\
        &=  \sum_{j=1}^n \sum_{i=1}^{n_j} \sum_{(k,q_{k,l}) \in H'(j,i)}  \pr(P_j > y_{j,i},  P_k > y_{k, l}) \cdot I_k(q_{k,l}, y_{k,l}).
    \end{align*}
\end{lemma}

\subsection{Gittins for Finite Support Distributions} \label{subsec: gittins-finite-support}

In this paper, we focus on the setting where every (predicted and true) job size distribution has finite support. This assumption implies that each (predicted or true) job size distribution has a finite number of quanta, which allows us to precompute all relevant quanta up-front (Algorithm \ref{alg: clairvoyant-alg}). It is likely true that our techniques can be extended to general distributions, but we focus on the finite support case to enable a clean algorithm design and analysis.

An intuitive property is that \gittins will only choose quanta for a job that line up with support points of that job, and that the largest support point coincides with the last -- of a finite number -- of the quanta. Thus, \gittins finishes all jobs. We prove these properties formally now.  

\begin{lemma}\label{prop: gittins_support}
    Let $\mathcal{I} = \{\dis_j\}_{j \in [n]}$ be given as input to  Algorithm \ref{alg: gittins-reformulated}. Let $y_{j,i}$ and $q_{j,i}$ be defined as above (Equation \ref{eq: quanta-length}). Then for all $j \in [n]$, quantum $(j,q_{j,i})$ can always be chosen so that $y_{j,i}$ is a support point of $\dis_j$ for all integer $2 \leq i \leq n_j$. Moreover, each job has finitely many quanta, so in particular, Algorithm \ref{alg: gittins-reformulated} finishes all jobs.
\end{lemma}

\begin{proof} 

    Consider any $j \in [n]$. Without loss of generality, we may assume $0$ is not the only support point of $\dis_j$. Let $P_j \sim \dis_j$.

    Let $q_{j,1}, \dots, q_{j,n_j}$ be the lengths of the quanta that Algorithm \ref{alg: gittins} runs on job $j$, in that order. Recall that $y_{j,i} = q_{j,1} + \cdots + q_{j,i-1}$, that is, $y_{j,i}$ is the attained processing time after running the $i-1$ quantum of job $j$. We want to show that we can choose the $q_{j,i}$'s (see Equation (\ref{eq: quanta-length})) so that the $y_{j,i}$'s are support points. We show that $y_{j,i}$ can be chosen to be a support point by induction on $i = 2, \dots, n_j$. 

    For the base case $i = 2$, we want to show there exists a support point $q$ maximizing $R_j(q, 0)$ over all $q \geq 0$. (Then we will take $y
    _{j,2} = q_{j,1} := q$, as desired.) Let $q' = \arg \max_{q \geq 0} R_j(q, 0)$. Note that since the numerator of $R_j(q,0)$ is equal to $\mathbb{P}(P_j \leq q \mid P_j > 0)$, in particular, since the probability is conditioned on the event that $P_j > 0$, we know that $q'$ must be at least the smallest nonzero support point (which exists by assumption). Now decreasing $q^*$ to the largest support point that is smaller than $q'$ only increases $R_j(q', 0)$ (because the numerator stays the same, and the denominator, $I_j(q',0)$, can only increase). Thus, we may take $q_{j,1}$ to be a support point, which implies that $y_{j,2} = q_{j,1}$ is also a support point. 
    
    Now we show $y_{j,i}$ can be chosen as a support point, given the inductive assumption that $y_{j,i-1}$ can be chosen as a support point. Suppose $q' = \arg \max_{q \geq 0} R_j(q, y_{j,i-1})$. We want to show there exists $q \geq 0$ such that $R_j(q, y_{j,i-1}) = R_j(q', y_{j,i-1})$ \textit{and} $y_{j, i-1} + q$ is a support point of job $j$'s distribution. (Then we may set $q_{j,i-1} = q$ and $y_{j,i} = y_{j, i-1} + q_{j,i-1}$ will be a support point, finishing the induction.) Choose $p$ so that $y_{j,i-1} \leq p \leq q'+y_{j,i-1}$, and $p$ is the largest support point between the upper and lower bounds; note such $p$ exists since, inductively, $y_{j, i-1}$ is a support point. Then $\mathbb{P}(P_j \leq p \mid P_j > y) = \mathbb{P}(P_j \leq y_{j,i-1} + q' \mid P_j > y)$, due to the maximality of $p$. On the other hand, $I_j(p - y_{j, i-1}, y_{j,i-1}) \leq I_j(q', y_{j, i-1})$. So taking $q = p-y_{j, i-1}$, we obtain that $R_j(q, y_{j,i-1}) \geq R_j(q', y_{j, i-1})$, and $q$ has the property that $y_{j, i-1} + q = y_{j,i}$ is a support point, as desired.

    Finally, we need to prove that $n_j$, the number of quanta of job $j$, is finite. Since the support of $\dis_j$ is finite, and we have shown that every $y_{j,i}$ is a support point, we just need to show that $q_{j,i} > 0$ for each $i$ (assuming $y_{j, i}$ is not the largest support point, in which case $R_j(q, y_{j,i})$ is vacuously 0). But this is clearly true, since if $q_{j,i} = 0$, then $R_j(q_i, y_{j,i}) = 0$, whereas we know there exists a $q$, namely the distance between $y_{j,i}$ and the next largest support point, such that $R_j(q, y_{j,i}) > 0$. Thus, $q_{j,i} =0$ contradicts that \textsf{GIPP} chooses quanta maximizing rank.

\end{proof}

\section{Gittins is not Robust to Distributional Errors}\label{sec: lower}

The goal of this section is to prove \Cref{thm: lower}, which we restate here for convenience. In short, naively running \gittins using input distributions with arbitrarily small errors can result in unbounded expected cost.

\thmlower*

The key idea of the proof is that quantum ranks are sensitive to small changes in the input distributions. We recall that the rank of quantum $(j,q)$ when $j$ has attained processing time $y$ is defined by
\[ R_j(q,y) = \frac{\mathbb{P}(P_j - y \leq q \mid P_j > y)}{I_j(q,y)}.\]
Assume that $\hat{\dis}_j$ is our predicted distribution with finite support, and $\hat{P}_j \sim \hat{\dis}_j$, and that $y$ is a support point of $\hat{\dis}_j$ and $y + q$ is the \textit{next smallest} support point. Then the rank of $(j,q)$ given prior processing time $y$  can be written as $\frac{\mathbb{P}(\hat{P}_j = y + q \mid \hat{P}_j > y)}{q}$.

Now we consider modifying $\hat{\dis}_j$ to obtain $\dis^*_j$ -- the true distribution of job $j$ -- by only moving the support point $y + q$ ``to the right'' by an arbitrarily small amount to $y + q + \varepsilon$ for $\varepsilon > 0$. Then, if we run \gittins using the ranks computed on the predicted quanta on the true distributions, we will run job $j$ for $q$ time units with the ``belief'' that we will invest $q$ units of processing to obtain expected reward $\mathbb{P}(\hat{P}_j = y + q \mid \hat{P}_j > y) > 0$. However, in reality, we obtain \emph{zero} expected reward, since $\mathbb{P}(P^*_j \leq y + q \mid P^*_j > y) = 0$.

Using this idea, we can trick \gittins into running quanta which will give much less reward than expected. This forces Gittins to delay the completion of many jobs by these less-valuable quanta. We now proceed formally to prove \Cref{thm: lower}.

\begin{proof}[Proof of \Cref{thm: lower}] 
    By the monotonicity of $\alpha$-closeness (\Cref{lem: dist_mono}), it suffices to prove the theorem for $\alpha \in (1,2)$. Consider any fixed $\alpha \in (1,2)$, and then define $\alpha = 1 + \varepsilon$, where $\varepsilon \in (0,1)$. Further, let $n \in \mathbb{N}$, $n \geq 2$. Then define the true instance $\mathcal{I}^*$ to be $n$ i.i.d. jobs distributed as 
    \[P^* = \begin{cases}
        1 + \varepsilon & w.p. \quad 1-p\\
        M & w.p. \quad p
    \end{cases},\]
    where we will choose $p$ and $M$ as functions of $n$ later. Similarly, we define $\hat{\mathcal{I}}$ to be $n$ i.i.d jobs distributed as
        \[\hat{P} = \begin{cases}
        1 & w.p. \quad 1-p\\
        M & w.p. \quad p
    \end{cases}.\]
    We quickly verify that $(\mathcal{I}^*, \hat{\mathcal{I}})$ are $(1+\varepsilon)$-close. To see this, consider the natural coupling between $P^*$ and $\hat{P}$ where $P^* = 1 + \varepsilon$ exactly when $\hat{P} = 1$, and likewise $P^* = M$ exactly when $\hat{P} = M$. By our coupling, we have $\frac{1}{1 + \varepsilon} P^* \leq \hat{P} \leq P^*$. It follows, for any $x \geq 0$, we have \[\pr(P^* > (1 + \varepsilon) x) \leq \pr(\hat{P} > x) \leq \pr(P^* > x),\]
    which is only stronger than the property that $(\mathcal{I}^*, \hat{\mathcal{I}})$ are $(1+\varepsilon)$-close.

    Next, we need to understand the behaviour of \gittins on the instance $\hat{\mathcal{I}}$. Since the jobs in $\hat{\mathcal{I}}$ are distributed i.i.d. as $\hat{P}$, which has finite support, there are only three relevant quanta to consider, by \Cref{prop: gittins_support}:
    \begin{itemize}
        \item A quantum of length $M$ at attained processing time $0$ (i.e. run the job once to completion). This has rank $\frac{1}{(1-p) + Mp}$.
        \item A quantum of length $1$ at attained processing time $0$. This has rank $\frac{1-p}{1}$.
        \item A quantum of length $M - 1$ at attained processing time $1$. This has rank $\frac{1}{M-1}$.
    \end{itemize}
    We take $p = \frac{1}{n}$ and $M = n^2$. This choice of $p$ and $M$ implies that $\frac{1}{(1-p) + Mp} \leq \frac{1 - p}{1}$, so \gittins first schedules all jobs in arbitrary order for $1$ unit of time each. Afterwards, for each such job that is still unfinished we run them in arbitrary order of $M-1$ time units each (until completion). Observe that if we follow this same algorithm -- but on the jobs in instance $\mathcal{I}^*$ instead -- then we are guaranteed to finish all jobs, since we process each job up until $M$ time units.

    It remains to lower bound the expected cost $\gittins(\mathcal{I}^*, \hat{\mathcal{I}})$ (Definition \ref{def: alg-cost}) of running \gittins on the true instance $\mathcal{I}^*$ \emph{using the quanta computed for $\hat{\mathcal{I}}$}, in terms of $\gittins(\mathcal{I}^*, \mathcal{I}^*) = \textsf{OPT}(\mathcal{I}^*)$ (\Cref{prop: gittins_opt}).  First, we run each of $n$ jobs for $1$ time unit. Since every job in $\mathcal{I}$ is supported only on $1+ \varepsilon$ and $M = n^2 > 1$, we complete \textit{no} jobs at this point. Then we run each of the $n$ jobs to completion in some fixed order. We crudely lower bound $\gittins(\mathcal{I}^*, \hat{\mathcal{I}})$ as follows: The number of jobs whose realized size is $M = n^2$ among the first $\lfloor \frac{n}{2} \rfloor$ in this fixed order is distributed as $Binom(\lfloor \frac{n}{2} \rfloor, \frac{1}{n})$. In particular, with probability $ 1- (1-1/n)^{\lfloor \frac{n}{2} \rfloor} = \Omega(1)$, there is at least one such long job among the first $\lfloor \frac{n}{2} \rfloor$. On this event, there are at least $\frac{n}{2}$-many jobs that complete after this long job. These jobs contribute at least $\frac{n}{2} \cdot n^2 = \Omega(n^3)$ to the expected total completion time, and since this event happens with constant probability, we conclude $\gittins(\mathcal{I}^*, \hat{\mathcal{I}}) = \Omega(n^3)$.

    On the other hand, we can upper bound $\opt(\mathcal{I}^*)$ by the following policy: (First phase) Run each job for $1 + \varepsilon$ time units in arbitrary order. (Second phase) Then run each unfinished job to completion in arbitrary order. Now we upper bound the expected total completion time of this policy. We say a job is \emph{short} if it has realized size $1 + \varepsilon$ and \emph{long} otherwise. In the first phase, we are guaranteed to complete all short jobs. Thus, the total completion times of all short jobs is at most $\sum_{k = 1}^n k \cdot (1 + \varepsilon) = O(n^2)$. After the first phase, we are at time $n \cdot (1 + \varepsilon)$, and only the long jobs remain. Let $L \sim Binom(n, \frac{1}{n})$ be the total number of long jobs. Then the expected total completion time of all long jobs is
    \begin{align*}
        \mathbb{E} \left[ (1 + \varepsilon)n \cdot L + \sum_{k = 1}^L (M - (1 + \varepsilon)) \cdot k \right] &\leq (1 + \varepsilon)n \cdot \mathbb{E}[ L ] + \mathbb{E} \left[ \sum_{k = 1}^L M \cdot k \right]\\
        &\leq (1 + \varepsilon)n \cdot \mathbb{E}[ L ] + M \cdot \mathbb{E} \left[ \frac{L(L+1)}{2} \right]\\
        &= O(n) + O(n^2) = O(n^2),
    \end{align*}
    where in the final step we use the fact that $\mathbb{E}[L] = 1$ and $\mathbb{E}[L^2] = O(1)$ since $L \sim Binom(n, \frac{1}{n})$. This gives the desired gap, $\frac{\gittins(\mathcal{I}^*, \hat{\mathcal{I}})}{\opt(\mathcal{I^*})} = \Omega(n)$.
\end{proof}
 
\section{Robust Gittins is Robust to Distributional Errors}
\label{sec:RG}

The goal of this section is to prove \Cref{thm:clairvoyant_main}, which we restate here for convenience. In subsection \ref{subsec:RGdefn} we design and define 
the Robust Gittins algorithm, and in subsection \ref{subsec:RGanalysis}
we analyze this algorithm. 

\thmupper*

\subsection{Definition of the Robust Gittins Algorithm}
\label{subsec:RGdefn}

The Robust Gittins Index Policy (\rg) is a modification of \gittins that avoids the pitfalls of naively running \gittins using quanta computed from the \textit{predicted} instance $\hat{\mathcal{I}}$ to define the scheduling policy on the true instance $\mathcal{I}^*$ (see Theorem \ref{thm: lower}). It is helpful to revisit the lower bound instance in the proof of Theorem \ref{thm: lower}. Notice that in this particular instance, we could have ``corrected'' the policy by running each quantum of length 1 for a tiny bit longer, namely, for $\alpha = 1+\varepsilon$ time units, where $\alpha$ is the error in the distributions. This is the idea behind \rg: the policy computes the ordered list of quanta from the predicted instance $\hat{\mathcal{I}}$ (recall we dubbed this ordering the \gittins order). \rg maintains the \gittins order but extends each quantum length by a factor of $\alpha$. \rg then runs these \textit{extended} quanta in the \gittins order. Now we describe \rg formally.

\begin{algorithm}[Robust Gittins Index Policy (\rg)] \label{alg: clairvoyant-alg}
The input initially provided to the policy  is the predicted instance $\hat{\mathcal{I}} = \{\hat{\dis}_j \}_{j \in [n]}$ and a parameter $\alpha$. There is a promise that $(\mathcal{I}^*, \hat{\mathcal{I}})$ are $\alpha$-close, where $\mathcal{I}^* = \{ \dis^*_j\}_{j \in [n]}$ is the true instance of distributions (\textit{not} known to the policy). All distributions (predicted and true) are finitely supported.

Initially the algorithm computes the quanta for $\hat{\mathcal{I}}$ as in Equation (\ref{eq: quanta-length}). For each $j \in [n]$, let $\hat{n}_j$ be the number of quanta for job $j$. For each $i \in [\hat{n}_j]$, let $(j, \hat{q}_{j,i})$ denote the $i$th quantum for job $j$, and $\hat{y}_{j,i}$ be the total processing time of the first $i-1$ quanta. 
Then the algorithm orders the  quanta $ \{(j, \hat{q}_{j,i}) \}_{j \in [n], i \in [\hat{n}_j]}$  in decreasing order of rank, i.e., in the \gittins order. 

The algorithm then schedules the quanta in the \gittins order, but runs each quantum $(j, \hat{q}_{j,i})$ for $\alpha \cdot \hat{q}_{j,i}$ units (instead of $\hat{q}_{j,i}$ units), or until job $j$ finishes, whichever comes first. 
\end{algorithm}

\subsection{Analysis: Proof of \Cref{thm:clairvoyant_main}}
\label{subsec:RGanalysis}

All notation is as in Theorem \ref{thm:clairvoyant_main}. We again will abuse notation slightly and use $\rg(\mathcal{I}^*,\hat{\mathcal{I}})$ to refer to the cost of Algorithm \ref{alg: clairvoyant-alg} as in Definition \ref{def: alg-cost}, as well as to the policy $\rg$ (Algorithm \ref{alg: clairvoyant-alg}) on the pair $(\mathcal{I}^*, \hat{\mathcal{I}})$ of true and predicted distributions, respectively.
First, we show that the hypothesis that $(\mathcal{I}^*, \hat{\mathcal{I}})$ are $\alpha$-close guarantees that  $\rg(\mathcal{I}^*,\hat{\mathcal{I}})$ finishes all jobs: that is, extending the length of the quanta computed according to $\hat{\mathcal{I}}$ by a factor of $\alpha$ guarantees that the sum of the extended quanta is at least the maximum support point for each job size distribution from $\mathcal{I}^*$. Thus, the policy's total completion time is well-defined / bounded.

\begin{lemma}\label{lem:clairvoyant_complete}
Let $(\mathcal{I}^*, \hat{\mathcal{I}})$ be an $\alpha$-close pair of families of finitely-supported distributions. Then $\rg(\mathcal{I}^*,\hat{\mathcal{I}})$ completes all jobs of $\mathcal{I}^*$.
\end{lemma}

\begin{proof}
    Let $\mathcal{I}^* = \{ \dis^*_j\}_{j \in [n]}$ and $\hat{\mathcal{I}} = \{\hat{\dis}_j \}_{j \in [n]}$. For each $j$, let $P_j^* \sim \mathcal{D}_j^*$ and $\hat{P}_j \sim \hat{\dis}_j$. It will helpful to consider the random variable $\alpha \cdot \hat{P}_j$. Let $x_{\text{max}}$ be the largest support point of $\alpha \cdot \hat{P}_j$ (which exists by the finite support assumption). The key observation is that the total length of the \textit{extended} quanta (as defined in Algorithm \ref{alg: clairvoyant-alg}) is equal to  $x_{\text{max}}$; this is because the total length of (unextended) quanta for $\hat{P}_j$ is equal to the largest support point of $\hat{P}_j$.

     Now, since $(\mathcal{I}^*, \hat{\mathcal{I}})$ is $\alpha$-close, 
    \[\pr(P_j^* > x_{\text{max}}) \leq \alpha \cdot \pr\left(\hat{P}_j > \frac{1}{\alpha} \cdot x_{\text{max}}\right) = \alpha \cdot \pr(\alpha \cdot \hat{P}_j > x_{\text{max}}) = 0\]
    where in the last equality we have used the definition of $x_{\text{max}}$. Since $\pr(P_j^* > x_{\text{max}}) = 0$, we have, using the key observation in the first paragraph, that $\rg(\mathcal{I}^*,\hat{\mathcal{I}})$ completes all jobs of $\mathcal{I}^*$.
\end{proof}

The next lemma is at the heart of our analysis. We will first upper bound the cost of \rg in terms of $\gittins(\hat{\mathcal{I}}, \hat{\mathcal{I}}) = \opt(\hat{\mathcal{I}})$, the optimal expected total completion time for instance $\hat{\mathcal{I}}$ among nonanticipatory policies. Note that the quantity we ultimately want to compare the cost of \rg against is not this, but rather, $\gittins(\mathcal{I}^*, \mathcal{I}^*) = \opt(\mathcal{I}^*)$. However, the following bound will be a key intermediate step.

\begin{lemma}\label{lem:clairvoyant_compare}
    Let $(\mathcal{I}^*, \hat{\mathcal{I}})$ be an $\alpha$-close pair of families of finitely-supported distributions. Then 

    \[\rg(\mathcal{I}^*, \hat{\mathcal{I}}) \leq \alpha^3 \cdot \gittins(\hat{\mathcal{I}}, \hat{\mathcal{I}}).\]
\end{lemma}

Before proving Lemma \ref{lem:clairvoyant_compare}, we will show how it swiftly implies Theorem \ref{thm:clairvoyant_main}. 

\begin{proof}[Proof of Theorem \ref{thm:clairvoyant_main}]
    Since $(\mathcal{I^*}, \hat{\mathcal{I}})$ is $\alpha$-close, by symmetry (Lemma \ref{lem: dist_sym}), $(\hat{\mathcal{I}}, \mathcal{I^*})$ is $\alpha$-close as well, so Lemma \ref{lem:clairvoyant_compare} implies both of the following statements:
    \begin{equation*}
        \rg(\mathcal{I}^*, \hat{\mathcal{I}}) \leq \alpha^3 \cdot \gittins(\hat{\mathcal{I}}, \hat{\mathcal{I}})
    \end{equation*}
    \begin{equation*}
        \text{and}
    \end{equation*}   
    \begin{equation*}
        \rg(\hat{\mathcal{I}}, \mathcal{I}^*) \leq \alpha^3 \cdot \gittins(\mathcal{I}^*, \mathcal{I}^*).
    \end{equation*}
Moreover, since $\rg(\hat{\mathcal{I}}, \mathcal{I}^*)$ is a scheduling policy for $\hat{\mathcal{I}}$ (that completes all jobs, by Lemma \ref{lem:clairvoyant_complete}), and $\gittins(\hat{\mathcal{I}}, \hat{\mathcal{I}})$ is optimal for $\hat{\mathcal{I}}$, we have

\[\gittins(\hat{\mathcal{I}}, \hat{\mathcal{I}}) \leq \rg(\hat{\mathcal{I}}, \mathcal{I}^*).\]

Combining the above inequalities, we have the desired bound:
\[\rg(\mathcal{I}^*, \hat{\mathcal{I}}) \leq \alpha^3 \cdot \gittins(\hat{\mathcal{I}}, \hat{\mathcal{I}}) \leq \alpha^3 \cdot \rg(\hat{\mathcal{I}}, \mathcal{I}^*) \leq \alpha^6 \cdot \gittins(\mathcal{I}^*, \mathcal{I}^*). \]

\end{proof}

Finally, we prove the key lemma, Lemma \ref{lem:clairvoyant_compare}

\begin{proof}[Proof of Lemma \ref{lem:clairvoyant_compare}]
Let $\mathcal{I}^* = \{ \dis^*_j\}_{j \in [n]}$ and $\hat{\mathcal{I}} = \{\hat{\dis}_j \}_{j \in [n]}$. For each $j$, let $P_j^* \sim \mathcal{D}_j^*$ and $\hat{P}_j \sim \hat{\dis}_j$.

We recall the closed-form expression for $\gittins(\hat{\mathcal{I}}, \hat{\mathcal{I}})$ (Lemma \ref{lem: closed-form-gipp}):

\begin{equation} \label{eq: closed-form-predicted}
    \gittins(\hat{\mathcal{I}}, \hat{\mathcal{I}}) = \sum_{j=1}^n \sum_{i=1}^{\hat{n}_j} \sum_{(k,q_{k,l}) \in \hat{H}'(j,i)} \pr(\hat{P}_j > \hat{y}_{j,i},\hat{P}_k > \hat{y}_{k,l}) \cdot \hat{I}_k(\hat{q}_{k,l}, \hat{y}_{k,l}),
\end{equation}
where $\hat{n}_j$, $\hat{y}_{j,i}$, and $\hat{q}_{j,i}$ are as in Algorithm \ref{alg: clairvoyant-alg}. Let $\{\hat{I}_k(\cdot, \cdot)\}_{k \in [n]}$ denote the investment functions for $\hat{\mathcal{I}}$ as in Equation (\ref{eq: investment-fn}).  Let $\hat{H}'(j,i)$ be as in Equation (\ref{eq: disjointified-quanta-prefixes}) for the instance $\hat{\mathcal{I}}$.

    We lower bound the outer double sum term-wise: fix $j \in [n]$ and $i \in [\hat{n}_j]$. Define 
    \[E_{j,i} := \sum_{(k,q_{k,l}) \in \hat{H}'(j,i)} \pr(\hat{P}_j > \hat{y}_{j,i},\hat{P}_k > \hat{y}_{k,l}) \cdot \hat{I}_k(\hat{q}_{k,l}, \hat{y}_{k,l}) \]
    
    Note that, by construction, the only quantum of job $j$ that is in $\hat{H}'(j,i)$ is $(j, q_{j,i})$. Thus,

    \begin{equation} \label{eq: Eij-rewritten}
        E_{j,i} = \pr(\hat{P}_j > \hat{y}_{j,i}) \cdot \hat{I}_j(\hat{q}_{j,i}, \hat{y}_{j,i}) \quad + \quad \pr(\hat{P}_j > \hat{y}_{j,i})  \cdot \sum_{\substack{(k,q_{k,l}) \in \\ \hat{H}'(j,i) \setminus (j,q_{j,i})}} \pr(\hat{P}_k > \hat{y}_{k,l}) \cdot \hat{I}_k(\hat{q}_{k,l}, \hat{y}_{k,l})
    \end{equation}
    where we have used the independence of $\hat{P}_j, \hat{P}_k$ for $j \neq k$ in the second term. 

    Now we lower bound the expression $\pr(\hat{P}_k > \hat{y}_{k,l}) \cdot \hat{I}_k(\hat{q}_{k,l}, \hat{y}_{k,l})$ for any $k,\ell$, including $k,\ell = j,i$.

    \begin{claim} \label{clm: term-wise-lower-bd}
        $\pr(\hat{P}_k > \hat{y}_{k,l}) \cdot \hat{I}_k(\hat{q}_{k,l}, \hat{y}_{k,l}) \geq \frac{1}{\alpha^2} \cdot \mathbb{E}\left[\min \{P_k^* - \alpha \cdot \hat{y}_{k, l},  \alpha \cdot \hat{q}_{k,l} \} \cdot \mathbf{1}_{\{P_k^* > \alpha \cdot \hat{y}_{k, l}\}}\right] $.
    \end{claim}

\begin{proof}[Proof of Claim \ref{clm: term-wise-lower-bd}]
Analogous to the above, let $\{I_k^*(\cdot, \cdot)\}_{k \in [n]}$ denote the investment functions for $\mathcal{I}^*$. We have 
    \begin{align}
        \pr(\hat{P}_k > \hat{y}_{k,l}) \cdot \hat{I}_k(\hat{q}_{k,l}, \hat{y}_{k,l}) &= \pr(\hat{P}_k > \hat{y}_{k,l}) \cdot \mathbb{E}[\min \{\hat{P}_k - \hat{y}_{k, l}, \hat{q}_{k,l} \} \mid \hat{P}_k > \hat{y}_{k, l}] \notag \\
        &= \mathbb{E}\left[\min \{\hat{P}_k - \hat{y}_{k, l}, \hat{q}_{k,l} \} \cdot \mathbf{1}_{\{\hat{P}_k > \hat{y}_{k, l}\}}\right] \label{eq: law-total-prob}\\
        &= \mathbb{E}\left[(\hat{P}_k - \hat{y}_{k, l}) \cdot \mathbf{1}_{\{\hat{P}_k \leq \hat{y}_{k,l} + \hat{q}_{k,l}\}} \mathbf{1}_{\{\hat{P}_k > \hat{y}_{k, l}\}}\right] + \mathbb{E}\left[\hat{q}_{k, l} \cdot \mathbf{1}_{\{\hat{P}_k > \hat{y}_{k,l} + \hat{q}_{k,l}\}} \mathbf{1}_{\{\hat{P}_k > \hat{y}_{k, l}\}}\right] \notag \\
        &=  \int_0^{\hat{q}_{k,l}} \pr(\hat{P}_k > \hat{y}_{k,l} + t) dt + \hat{q}_{k,l} \cdot \pr(\hat{P}_k > \hat{y}_{k,l} + \hat{q}_{k,l}) \label{eq: integrals-expectation} \\
        &\geq \frac{1}{\alpha} \cdot \int_0^{\hat{q}_{k,l}} \pr(P_k^* > \alpha \cdot (\hat{y}_{k,l} + t)) dt + \frac{1}{\alpha} \cdot \hat{q}_{k,l} \cdot \pr(P_k^* > \alpha \cdot (\hat{y}_{k,l} + \hat{q}_{k,l})) \label{eq: alpha-close-lower-bd}\\
        &= \frac{1}{\alpha} \cdot \left( \frac{1}{\alpha} \cdot\int_0^{\alpha \cdot \hat{q}_{k,l}} \pr(P_k^* > \alpha \cdot \hat{y}_{k,l} + t) dt +  \hat{q}_{k,l} \cdot \pr(P_k^* > \alpha \cdot(\hat{y}_{k,l} + \hat{q}_{k,l})) \right) \label{eq: u-sub} \\
        &= \frac{1}{\alpha^2} \cdot \left( \int_0^{\alpha \cdot \hat{q}_{k,l}} \pr(P_k^* > \alpha \cdot \hat{y}_{k,l} + t) dt + \alpha \cdot \hat{q}_{k,l} \cdot \pr(P_k^* > \alpha \cdot (\hat{y}_{k,l} + \hat{q}_{k,l})) \right) \notag \\
        &= \frac{1}{\alpha^2} \cdot \mathbb{E}\left[\min \{P_k^* - \alpha \cdot \hat{y}_{k, l},  \alpha \cdot \hat{q}_{k,l} \} \cdot \mathbf{1}_{\{P_k^* > \alpha \cdot \hat{y}_{k, l}\}}\right] \label{eq: term2-lower} 
    \end{align}

where in (\ref{eq: law-total-prob}) we have used the Law of Total Expectation, in (\ref{eq: integrals-expectation}) the formula for expected value in terms of integrals, in (\ref{eq: alpha-close-lower-bd}) the $\alpha$-close hypothesis, and in (\ref{eq: u-sub}) a $u$-substitution.
\end{proof}

Continuing from Equation (\ref{eq: Eij-rewritten}), we have:
\begin{align}
     E_{j,i} &= \pr(\hat{P}_j > \hat{y}_{j,i}) \cdot \hat{I}_j(\hat{q}_{j,i}, \hat{y}_{j,i}) + \pr(\hat{P}_j > \hat{y}_{j,i})  \cdot \sum_{\substack{(k,q_{k,l}) \in \\ \hat{H}'(j,i) \setminus (j,q_{j,i})}} \pr(\hat{P}_k > \hat{y}_{k,l}) \cdot \hat{I}_k(\hat{q}_{k,l}, \hat{y}_{k,l}) \notag \\
     &\geq \frac{1}{\alpha^2} \cdot \mathbb{E}\left[\min \{P_j^* - \alpha \cdot \hat{y}_{j, i},  \alpha \cdot \hat{q}_{j,i} \} \cdot \mathbf{1}_{\{P_j^* > \alpha \cdot \hat{y}_{j, i}\}}\right] \notag \\
     &\quad + \quad \frac{1}{\alpha} \cdot \pr(P_j^* > \alpha \cdot \hat{y}_{j,i}) \cdot \sum_{\substack{(k,q_{k,l}) \in \\ \hat{H}'(j,i) \setminus (j,q_{j,i})}} \frac{1}{\alpha^2} \cdot \mathbb{E}\left[\min \{P_k^* - \alpha \cdot \hat{y}_{k, l},  \alpha \cdot \hat{q}_{k,l} \} \cdot \mathbf{1}_{\{P_k^* > \alpha \cdot \hat{y}_{k, l}\}}\right] \label{eq: replacements} \\
     &\geq \frac{1}{\alpha^3} \cdot \mathbb{E}\left[\min \{P_j^* - \alpha \cdot \hat{y}_{j, i},  \alpha \cdot \hat{q}_{j,i} \} \cdot \mathbf{1}_{\{P_j^* > \alpha \cdot \hat{y}_{j, i}\}}\right] \notag \\
     &\quad + \quad \frac{1}{\alpha^3} \cdot  \sum_{\substack{(k,q_{k,l}) \in \\ \hat{H}'(j,i) \setminus (j,q_{j,i})}} \pr(P_j^* > \alpha \cdot \hat{y}_{j,i}) \cdot \mathbb{E}\left[\min \{P_k^* - \alpha \cdot \hat{y}_{k, l},  \alpha \cdot \hat{q}_{k,l} \} \cdot \mathbf{1}_{\{P_k^* > \alpha \cdot \hat{y}_{k, l}\}}\right] \notag \\
     &= \frac{1}{\alpha^3} \cdot \mathbb{E}\left[\min \{P_j^* - \alpha \cdot \hat{y}_{j, i},  \alpha \cdot \hat{q}_{j,i} \} \cdot \mathbf{1}_{\{P_j^* > \alpha \cdot \hat{y}_{j, i}\}}\right] \notag \\
     &\quad + \quad \frac{1}{\alpha^3} \cdot  \sum_{\substack{(k,q_{k,l}) \in \\ \hat{H}'(j,i) \setminus (j,q_{j,i})}} \mathbb{E}[\mathbf{1}_{\{P_j^* > \alpha \cdot \hat{y}_{j,i})\}}] \cdot \mathbb{E}\left[\min \{P_k^* - \alpha \cdot \hat{y}_{k, l},  \alpha \cdot \hat{q}_{k,l} \} \cdot \mathbf{1}_{\{P_k^* > \alpha \cdot \hat{y}_{k, l}\}}\right] \notag \\
     &= \frac{1}{\alpha^3} \cdot \mathbb{E}\left[\min \{P_j^* - \alpha \cdot \hat{y}_{j, i},  \alpha \cdot \hat{q}_{j,i} \} \cdot \mathbf{1}_{\{P_j^* > \alpha \cdot \hat{y}_{j, i}\}}\right] \notag \\
     &\quad + \quad \frac{1}{\alpha^3} \cdot  \sum_{\substack{(k,q_{k,l}) \in \\ \hat{H}'(j,i) \setminus (j,q_{j,i})}} \mathbb{E}\left[\min \{P_k^* - \alpha \cdot \hat{y}_{k, l},  \alpha \cdot \hat{q}_{k,l} \} \cdot \mathbf{1}_{\{P_k^* > \alpha \cdot \hat{y}_{k, l}\}} \cdot \mathbf{1}_{\{P_j^* > \alpha \cdot \hat{y}_{j,i}\}}\right] \label{eq: merge-expectation} \\
     &= \frac{1}{\alpha^3} \cdot \sum_{(k,q_{k,l}) \in  \hat{H}'(j,i)} \mathbb{E}\left[\min \{P_k^* - \alpha \cdot \hat{y}_{k, l},  \alpha \cdot \hat{q}_{k,l} \} \cdot \mathbf{1}_{\{P_k^* > \alpha \cdot \hat{y}_{k, l}\}} \cdot \mathbf{1}_{\{P_j^* > \alpha \cdot \hat{y}_{j,i}\}}\right] \label{eq: Eij-final-bd} 
\end{align}
where in (\ref{eq: replacements}) we have used both Claim \ref{clm: term-wise-lower-bd} and the bound $\pr(\hat{P}_j > \hat{y}_{j,i}) \geq \frac{1}{\alpha} \cdot \pr(P_j^* > \alpha \cdot \hat{y}_{j,i})$ from the $\alpha$-close hypothesis, and in (\ref{eq: merge-expectation}) we have used that $P_k^*$ and $P_j^*$ are independent for $j \neq k$ (recalling also that the only quantum of job $j$ that is in $\hat{H}'(j,i)$ is $(j, q_{j,i})$). 
    
     Substituting (\ref{eq: Eij-final-bd}) into (\ref{eq: closed-form-predicted}), we conclude the proof:
     \begin{align*}
        \gittins(\hat{\mathcal{I}}, \hat{\mathcal{I}}) &\geq \frac{1}{\alpha^3} \cdot  \sum_{j=1}^n \sum_{i=1}^{n_j} \sum_{(k,l) \in \hat{H}'(j,i)} \mathbb{E}\left[\min \{P_k^* - \alpha \cdot \hat{y}_{k, l},  \alpha \cdot \hat{q}_{k,l} \} \cdot \mathbf{1}_{\{P_k^* > \alpha \cdot \hat{y}_{k, l}\}} \cdot \mathbf{1}_{\{P_j^* > \alpha \cdot \hat{y}_{j,i}\}}\right] \\
        &= \frac{1}{\alpha^3} \cdot \rg(\mathcal{I}^*, \hat{\mathcal{I}}).
     \end{align*}
The last equality follows line-by-line from the proof of Lemma \ref{lem: closed-form-gipp} (Lemma 2.1 in \cite{MegowV14}). In particular, the proof gives the same closed-form expression for the scheduling policy's expected cost for \textit{any} quanta lengths (not necessarily the ones yielding an optimal policy). The only requirement is that the quanta lengths of a job sum to at least the largest support point, i.e., that the jobs finish. This is satisfied in our case, by Lemma \ref{lem:clairvoyant_complete}.
\end{proof}

\section{Conclusion}

This paper introduced a new stochastic scheduling model that captures errors in distributions given to the policy. We devised a policy  that is provably robust to errors in reported distributions for perhaps the simplest and most widely studied stochastic scheduling problem (minimizing the expected mean completion time of stochastic jobs on a single machine with preemption), giving an answer to the open question of Scully, Grosof, and Mitzenmacher \cite{scully_et_al:LIPIcs.ITCS.2022.114} for the no-arrivals setting with non-indentically distributed jobs (rather than the M/G/1 queue). The performance of our policy is parametrized in terms of a distance measure of our design, and our policy approaches the optimal as the error goes to $1$ (no error). We emphasize that we make minimal assumptions on the scheduling environment: the job size distributions have finite support and are independent.

This model initiates a new perspective on coping with erroneous distributions to provide robust stochastic optimization algorithms. Algorithms for many stochastic optimization problems are sensitive, requiring the distributions given to the algorithm to be precisely correct.  It is of interest to apply the model considered in this paper more broadly in stochastic optimization. Some especially interesting open questions are:
\begin{itemize}
    \item Our policy needs to know an upper bound on the error, $\alpha$, up-front. Can we design a policy that does not need to know $\alpha$? The main challenge in this harder setting is that the scheduler only observes a \emph{single sample} of the true distribution (and, as in our setting, does not know the job size until the job is done processing). In particular, the scheduler cannot necessarily differentiate between errors in the predicted distributions and randomness in the true distributions.
    \item While we consider one formulation of Gittins applied to a particular scheduling problem, the Gittins index policy is optimal for a wide range of stochastic decision-making problems (e.g. multi-armed bandits). Can we generalize our policy and error measure in this paper to other problems where Gittins  is optimal? In particular, this paper address the case with no arrivals. It would be interesting to consider the case where there are arrivals, e.g., stochastic arrivals as in the M/G/1 queue.  The analytical techniques in this paper leverage properties of the Gittins index in the no arrival case and we expect the case with arrivals to require new ideas. 
    \item This paper offers a new model for how to make algorithms robust to erroneous distributions.  It is of interest to better understand this and related models for stochastic optimization problems.  
\end{itemize}

It is plausible that our idea of adopting a symmetric distance measure (as contrasted with, e.g., the popular distance measure of KL divergence),
and then using this symmetry between the predicted and true distributions
to enable a clean analysis, could find application in addressing such open questions. 

\printbibliography

\end{document}